\documentclass[11pt,times]{article}
\usepackage[a4paper,margin=2cm]{geometry}

\usepackage{latexsym}
\usepackage{epsfig}
\usepackage{times}
\usepackage{amssymb}
\usepackage{enumerate}
\usepackage{bm}
\usepackage{enumitem}
\usepackage{mathtools}
\usepackage{xfrac}
\usepackage{hyperref}
\usepackage{amsthm}
\usepackage{tikz}
\usetikzlibrary{arrows.meta, decorations.pathreplacing}
\usepackage{tikz-3dplot}
\usepackage{caption}
\usepackage{subcaption}
\usepackage{csquotes}

\usepackage[ruled,linesnumbered]{algorithm2e}

\usepackage{tikz}

\usepackage{amsthm, amsmath,amsfonts, amssymb}

\usepackage{setspace}
\newcommand{\ignore}[1]{{}}

\def\p1{\phantom{1}}

\newtheorem{theorem}{Theorem}
\newtheorem{corollary}{Corollary}
\newtheorem{lemma}{Lemma}

\newtheorem{fact}{Fact}

\DeclareMathOperator{\EX}{\mathbb{E}}
\DeclareMathOperator*{\argmax}{arg\,max}

\newcommand{\full}[1]{{}}

\begin{document}

\begin{titlepage}

\title{Online $b$-Matching with Stochastic Rewards
}
\author{Susanne Albers\thanks{Department of Computer Science, Technical University of Munich. {\tt albers@in.tum.de}} \and Sebastian Schubert\thanks{Department of Computer Science, Technical University of Munich. {\tt sebastian.schubert@tum.de}}}
\maketitle

\thispagestyle{empty}

\begin{abstract}
The $b$-matching problem is an allocation problem where the vertices on the
left-hand side of a bipartite graph, referred to as servers, may be matched
multiple times. In the setting with stochastic rewards, an assignment between
an incoming request and a server turns into a match with a given success
probability. Mehta and Panigrahi (FOCS 2012) introduced online bipartite 
matching with stochastic rewards, where each vertex may be matched once.
The framework is equally interesting in graphs with vertex capacities.
In Internet advertising, for instance, the advertisers seek successful
matches with a large number of users. 

We develop (tight) upper and lower bounds on the competitive ratio of 
deterministic and randomized online algorithms, for $b$-matching with
stochastic rewards. Our bounds hold for both offline benchmarks considered 
in the literature. As in prior work, we first consider vanishing 
probabilities. We show that no randomized online algorithm can achieve 
a competitive ratio greater than $1-1/e\approx 0.632$, even for 
identical vanishing probabilities and arbitrary uniform server capacities.
Furthermore, we conduct a primal-dual analysis of the deterministic \textsc{StochasticBalance} algorithm. We prove that it achieves a competitive
ratio of $1-1/e$, as server capacities increase, for arbitrary heterogeneous 
non-vanishing edge probabilities. This performance guarantee 
holds in a general setting where servers have individual capacities and for 
the vertex-weighted problem extension. To the best of our knowledge, this
is the first result for \textsc{StochasticBalance} with arbitrary non-vanishing
probabilities.

We remark that our impossibility result implies in particular that, for the 
AdWords problem, no online algorithm can be better than $(1-1/e)$-competitive 
in the setting with stochastic rewards. 
\end{abstract}

\end{titlepage}

\section{Introduction}
Online bipartite matching and the generalized AdWords problem have received 
tremendous research interest over the last three decades, see e.g.~\cite{AGKM,BM,BJN,BNW,GU,HT23,HZ,HZZ,KVV,M13,MP,MSVV,U23}. 
In this realm Mehta and Panigrahi~\cite{MP} introduced the online matching problem with stochastic rewards, which focuses on \emph{successful}
allocations. Given a bipartite graph, the left-hand side vertices are known in advance. The right-hand side vertices 
arrive one by one. In each step, an incoming vertex may be assigned to an available neighbor. Then, with a certain 
probability, the assignment turns into an actual match. Mehta and Panigrahi motivate their study by the fact that
in Internet advertising, an advertiser only pays if a user clicks on the assigned ad (pay-per-click) and click-through
rates are known. 

Further and recent work on matchings with stochastic rewards includes~\cite{GNR,GU,HJSS,HZ,MWZ,U24}. Almost all the contributions 
address the uncapacitated setting where each vertex may be matched once. However, matchings with stochastic rewards are also 
highly relevant in capacitated environments. In fact, in the AdWords problem, the left-hand side vertices of the bipartite graph 
represent advertisers with budgets who wish to show ads to a large number of users. Furthermore, the left-hand side vertices 
could be service providers that offer support in terms of telecommunications, data storage or job processing to a large
number of clients. A client with a service request may accept an offer with a certain probability. Alternatively, the left-hand
side vertices could be online retailers that sell large quantities of products. A potential customer buys a recommended product 
with a certain chance. 

\vspace*{0.1cm}

{\bf Problem definition:} In this paper we study online $b$-matching with stochastic rewards. The $b$-matching problem is a 
well-known capacitated allocation framework. In particular, it models interesting special cases of the AdWords problem. Formally, 
we are given a bipartite graph $G=(S\cup R,E)$. The vertices of $S$ are servers. Each server $s\in S$ has a capacity of $b$,
indicating that may it be matched up to $b$ times. The vertices of $R$ are requests that have to be assigned to the servers. 
The set $S$ of servers is known in advance. The requests of $R$ arrive online, one by one. Whenever a new request 
$r\in R$ arrives, its incident edges are revealed. Each such edge $\{s,r\}$ has a success probability $p_{s,r}$ that is 
revealed as well. An algorithm has to assign the request immediately and irrevocably to an eligible server with remaining 
capacity, provided that there is one. If an assignment is made, request $r$ accepts it with the probability of the corresponding 
edge, i.e., the assignment turns into a successful match. The outcome of the random choice is independent of past ones. 
If $r$ does not accept, it leaves the system and the remaining capacity of the proposed server is unchanged. The goal is to 
maximize the expected number of successful matches. Throughout this paper we will use the term \emph{assignment} to refer
to the decision of an algorithm. The terms \emph{success} or \emph{match} are used for assignments that succeeded. 

\vspace*{0.1cm}

{\bf Benchmarks and competitiveness:} We evaluate the performance of online algorithms using competitive analysis. 
For online matching with stochastic rewards, addressing the uncapacitated variant with $b=1$, there has been some discussion
in the literature~\cite{GU,MP} on how to quantify the value of an optimal solution.  When Mehta and Panigrahi~\cite{MP} introduced the problem, 
they compared their algorithms against the offline and non-stochastic optimum. More specifically, 
they consider the optimal solution to the problem where the entire graph is known in advance and the reward of adding
an edge $\{s, r\}$ to the matching is (deterministically) $p_{s,r}$. The reward that can be accrued per server is upper 
bounded by $1$. The problem of maximizing total reward is known as the budgeted allocation problem. Mehta and Panigrahi
show that the expected number of matches of any online algorithm for the matching problem with stochastic rewards is 
upper bounded by the value of a (fractional) optimal solution to the budgeted allocation problem. This fact immediately 
carries over to the capacitated problem, assuming that the accrued reward per server is upper bounded by its capacity $b$.
Hence the fractional optimum to the budgeted allocation problem is a suitable benchmark for evaluating an online
algorithm. In the following, we will refer to it as the \emph{non-stochastic benchmark}.

The other benchmark that is used in the literature is referred to as the 
\emph{clairvoyant} or simply \emph{stochastic benchmark}. It corresponds to an 
offline algorithm assigning requests. More specifically, the benchmark is 
defined as the expected number of matches produced by an optimal algorithm that 
knows the entire graph including the edge probabilities in advance. The 
algorithm has to assign requests according to their arrival order and is only 
informed after an assignment whether it succeeded or failed. Golrezaei et al.~\cite{GNR} showed that, for any graph, the stochastic benchmark is upper 
bounded by the non-stochastic benchmark. Hence the stochastic benchmark 
potentially admits better competitive ratios.

Given an input graph $G$, let $\mathbb{E}[\textsc{Alg}(G)$] denote the expected 
number of matches of an online algorithm {\sc Alg} for the online $b$-matching 
problem with stochastic rewards. Let $\textsc{Opt}(G)$ and $\textsc{SOpt}(G)$ be 
the value of the non-stochastic and stochastic benchmark on $G$, respectively. 
Algorithm $\textsc{Alg}$ is $c$-competitive against the non-stochastic benchmark 
if ${\mathbb{E}[\textsc{Alg}(G)]}\geq c\cdot {\textsc{Opt}(G)}$ holds, for all $G$. 
Analogously, {\sc Alg} is $c$-competitive against the 
stochastic benchmark if
${\mathbb{E}[\textsc{Alg}(G)]}\geq c\cdot {\textsc{SOpt}(G)}$, for all $G$.

\subsection{Related Work}
First, we review a few results on online matching \emph{without} stochastic 
rewards, i.e.\ an assignment is a match (with probability~1). Online bipartite 
matching was introduced in a seminal paper by Karp et al.~\cite{KVV}. Each vertex may 
be matched once. Karp et al.\ showed that the best competitive ratio of deterministic 
online algorithms is equal to $1/2$. Furthermore, they proposed a randomized 
{\sc Ranking} algorithm that achieves an optimal competitiveness of 
$1-1/e \approx 0.632$, see~\cite{BM,DJK,KVV}. Online $b$-matching was studied
in~\cite{KP}, assuming that all servers have a uniform capacity of $b$. 
The best competitive ratio of deterministic online algorithms is equal to
$1-1/(1+1/b)^b$, which tends to $1-1/e$ as $b\rightarrow \infty$~\cite{KP}. 

Online matching with stochastic rewards was defined by Mehta and Panigrahi~\cite{MP}.
Again, in the original problem each vertex may be matched only once, i.e.\ $b=1$.
Mehta and Panigrahi consider equal success probabilities $p_{s,r}=p$,
for all edges $\{s,r\} \in E$. All their 
results are in comparison with \textsc{Opt}. 
Mehta and Panigrahi present a simple \textsc{StochasticBalance} algorithm that assigns incoming requests to an eligible neighbor that is assigned the least amount of requests so far. 
They show that its competitive ratio is at least $0.567$ and at most $0.588$ for vanishing probabilities, meaning that $p \rightarrow 0$. Moreover, they analyze \textsc{Ranking} and show that it achieves 
a competitive ratio of at least $0.534$, for equal non-vanishing probabilities. Finally, they show that no (randomized) algorithm has a competitive ratio greater than $0.621 < 1-1/e$. 
This upper bound has recently been improved to $0.597$ by leveraging reinforcement learning \cite{ZSZJD}.

Huang and Zhang~\cite{HZ} were the first to successfully 
conduct a primal-dual 
analysis of the problem. They improve the competitiveness of \textsc{StochasticBalance} against \textsc{Opt} to $0.576$ if all edges have 
equal and vanishing probabilities. Furthermore, they also show that a generalized version of \textsc{StochasticBalance} is at least $0.572$-competitive if all 
edges have unequal vanishing probabilities. 
Concurrently, Goyal and Udwani \cite{GU} analyzed the performance of a randomized \textsc{PerturbedGreedy} algorithm against the stochastic benchmark\footnote{In fact, Goyal and Udwani consider an even stronger benchmark than \textsc{SOpt}, where the offline algorithm is allowed to choose the arrival order of the requests.} \textsc{SOpt}, for the more general vertex-weighted problem.
They show a competitive ratio of $1-1/e$ if the edge probabilities are \emph{decomposable}, i.e., they
are the product of two factors, one for each of the two vertices
of an edge. 
For unequal but vanishing probabilities, they give an algorithm that is $0.596$-competitive. In the most recent work, Huang et al. \cite{HJSS} improve the competitiveness of \textsc{Ranking} against \textsc{Opt} to $0.572$, if 
all probabilities are equal and vanishing. They further show that \textsc{StochasticBalance} is $0.613$-competitive against \textsc{SOpt} in the case of equal and vanishing probabilities. Their result slightly worsens to $0.611$ if 
the probabilities are unequal, but still vanishing. 

We remark that the lower bounds, i.e.\ the competitive ratios of the various algorithms mentioned in the last two paragraphs also hold for general
$b$-matching with stochastic rewards. This follows from a standard 
vertex-splitting argument: Replace each server of capacity $b$ by $b$ vertices
that may be matched only once. On the resulting graph, execute the corresponding
algorithm. However, the upper bounds cannot immediately be extended to other 
values of $b$. 
The only prior work that covers online $b$-matching with stochastic rewards is
by Golrezaei et al.~\cite{GNR}. The authors study more general personalized assortments optimization in management science. Their \textsc{InventoryBalance} algorithm achieves a competitive ratio of $1-1/e$ against \textsc{Opt} 
for arbitrary edge probabilities, if $b\rightarrow \infty$. They also show that this is tight, by giving an upper bound 
of $1-1/e$ if all edges have probability $1$. For  $b=1$, their algorithm reduces 
to a greedy algorithm that is $1/2$-competitive. 

We finally mention the Adwords problem~\cite{MSVV}, which has been studied without stochastic rewards. There is a set of advertisers, each with a daily budget, who wish to link their ads to search keywords and issue respective bids. Queries along with 
their keywords arrive online and must be allocated to the advertisers. The
optimal competitiveness is $1-1/e$, under the small-bids assumption~\cite{BJN,MSVV}. 
The $b$-matching problem models the basic setting where each advertiser issues
bids of value~0 or 1. 

\subsection{Our Contribution}
In this paper we develop tight upper and lower bounds on the competitive ratio
of deterministic and randomized online algorithms for $b$-matching with stochastic
rewards. Our bounds hold for both benchmarks \textsc{Opt} and \textsc{SOpt}.
Our study is specifically motivated by the following question:
Is it possible to beat the barrier of $1-1/e$, for competitiveness,
in the (easiest) setting of equal and vanishing probabilities when 
comparing against the stochastic benchmark? Observe that, for $b=1$ and vanishing probabilities,
deterministic online algorithms achieve competitive ratios greater than
$1/2$, which is the best bound in the framework without stochastic 
rewards. The question was also raised by Goyal and Udwani \cite{GU}. Vanishing probabilities are sensible
in various applications. In Internet advertising, a user clicks on an assigned ad with low probability.

In Section~\ref{sec:ub} we develop an upper bound. Here we consider vanishing
probabilities, as almost all prior work. We prove that no
randomized online algorithm can achieve a competitive ratio greater than
$1-1/e$ against the stochastic benchmark \textsc{SOpt}, for online
$b$-matching with stochastic rewards. This holds even for equal vanishing
edge probabilities and for all values of $b$. 
To the best of our knowledge, this is the first hardness 
result for the stochastic benchmark in the literature. We immediately
obtain the same upper bound against the non-stochastic benchmark \textsc{Opt}
because, as stated above, for any graph the value of \textsc{Opt} is at least
as large as that of \textsc{SOpt}~\cite{GNR}. In conclusion, surprisingly, it is impossible
to break the barrier of $1-1/e$, which is a recurring performance guarantee in
matching problems. In particular for Adwords, a prominent application of 
matchings with stochastic rewards, no further improvement is possible. 

Technically, in our upper bound construction, we define a family of graphs,
for general $n=|S|$ and $b$. It generalizes graphs used by Mehta and
Panigrahi~\cite{MP} to prove their upper bound of $0.621 < 1-1/e$. 
For our problem, the mathematical analysis is much more involved. A key problem is to 
estimate the value of the stochastic benchmark \textsc{SOpt}.  
An offline algorithm is hindered
by the fact that it does not know in advance which edges will be successful
and that it has to serve requests in the order of arrival. We resolve
this issue by considering a \textsc{Greedy} strategy, for the graph family. 

In Section~\ref{sec:bal} we present a constructive, algorithmic result. We analyze 
a generalization of the deterministic 
\textsc{StochasticBalance} algorithm, for online $b$-matching with stochastic rewards. 
Specifically, we perform a primal-dual analysis, considering the harder
benchmark \textsc{Opt}. In order to ease the technical exposition, we first 
assume that all servers have a uniform capacity of $b$. We prove that 
\textsc{StochasticBalance} achieves a competitive ratio of $1-1/e$ with respect
to \textsc{Opt}, as $b\rightarrow \infty$. This bound holds for arbitrary
individual, non-vanishing edge probabilities. The performance guarantee immediately 
carries over to the easier stochastic benchmark \textsc{SOpt}. Finally,
we show how to extend the result to more general settings: 
(1)~Each server $s\in S$ has an individual capacity $b_s$. (2) Each server 
$s\in S$ has a weight $w_s$ and any successful match incident to $s$ 
has a weight/value of $w_s$. The goal is to maximize the expected total 
weight of successful matches. \textsc{StochasticBalance} remains $(1-1/e)$-competitive
as the minimum server capacity increases, i.e., $\min_{s\in S} b_s \rightarrow \infty$. 
We remark that, due to the complexity of the primal-dual analysis, we present bounds as $b$ or $\min_{s\in S} b_s$ tend to infinity. Already for $b=1$, the analysis of \textsc{StochasticBalance}
involves solving integral equations using numerical methods~\cite{HJSS,HZ}.

Our contribution for \textsc{StochasticBalance} is
one of the few existing results that hold for arbitrary non-vanishing probabilities. 
\textsc{StochasticBalance} was analyzed for identical non-vanishing probabilities, the best ratio
of $0.567$ being achieved as $p\rightarrow 0$~\cite{MP}.
The randomized algorithms {\sc Ranking} and {\sc PerturbedGreedy}
attain competitive ratios of $0.572$ and $1-1/e$ for equal or decomposable,
non-vanishing probabilities~\cite{GU,HJSS}. 
Finally, in comparison, \textsc{StochasticBalance} appears to be a more favourable 
algorithm than \textsc{InventoryBalance}, for online $b$-matching with stochastic rewards. 
It is a simpler, more intuitive and well studied matching algorithm~\cite{HJSS,HZ,MP}, 
attaining competitive ratios greater than $1/2$, for $b=1$.

\section{Upper Bound of $1-1/e$ for Vanishing Probabilities} 
\label{sec:ub}
We start by giving an upper bound on the competitiveness of any (randomized) algorithm for the online $b$-matching problem with 
stochastic rewards. Note that for $p=1$, both benchmarks are identical, and it is known \cite{MSVV} that no 
algorithm achieves a competitive ratio greater than $1-1/e$. However, 
it was unknown if one could improve upon this in the case of smaller edge probabilities or when using the stochastic benchmark. 
We will show that this is impossible. To the best of our knowledge, we develop the first hardness result with respect to the stochastic benchmark.
As mentioned before, the stochastic benchmark is easier than the non-stochastic one. 
Hence we immediately obtain the same upper bound for the non-stochastic benchmark. 

We investigate the family of graphs $G_n^b$ with $n$ servers 
$S = \{s_1, s_2, \ldots, s_n \}$ of capacity $b$ and $n\cdot b/p$ requests. 
The online side is divided into $n$ rounds, each containing $b/p$ identical requests. 
All requests of round $i$, $1\leq i \leq n$, are connected to servers $\{s_i, s_{i+1}, \ldots, s_n\}$. 
The success probability of every edge is $p \rightarrow 0$. 

Note that Mehta and Panigrahi \cite{MP} use the family $G_n^1$ to show their upper bound of $0.621 < 1-1/e$ for $b=1$ against \textsc{Opt}. 
They first prove that \textsc{StochasticBalance} is an optimal online algorithm for $G_n^1$, for all $n$. 
In fact, their proof also works for all values of $b$. Afterward, they compute the exact expected number of successes generated by \textsc{StochasticBalance} 
with input $G_n^1$ for $n=1, 2$ and $3$. The competitiveness of $0.621$ against \textsc{Opt} is achieved on $G_3^1$. 

\begin{lemma}[cf., Lemma 12 of Mehta and Panigrahi \cite{MP}]
\label{lem:MP}
    \textsc{StochasticBalance} is optimal for input graph $G_n^b$, for any combination of $n$ and $b$. 
\end{lemma}

We will show our upper bound similarly. In Section \ref{sec:ubBAL},
we provide an upper bound for the expected number of successful matches of \textsc{StochasticBalance} on $G_n^b$. 
By Lemma \ref{lem:MP}, this 
upper bounds the performance of any online algorithm. 
Afterward, in Section \ref{sec:ubSOPT}, we quantify 
$\textsc{SOpt}\left(G_n^b\right)$ and compare it to the upper bound 
from the previous section. 
Note that is much more involved than determining $\textsc{Opt}\left(G_n^b\right)$, 
which is simply $n\cdot b$. 

\subsection{Expected Number of Matches by \textsc{StochasticBalance} on $G_n^b$}
\label{sec:ubBAL}
In this section, we develop a simple upper bound for the expected number of 
matches generated by \textsc{StochasticBalance} on $G_n^b$. 
We resort to an upper bound, since the exact expected number of matches  
becomes very difficult to determine for arbitrary values of $n$ and $b$. 
However, one can show that our upper bound becomes tight if $n \rightarrow \infty$, which we will use to derive the hardness result in the end. 

For equal edge probabilities, \textsc{StochasticBalance} simply assigns an 
incoming request $r$ to a neighbor with remaining capacity that has 
minimum \emph{load}. The load of a server $s$ is defined as
the sum of edge probabilities to requests assigned to $s$. 
Observe that \textsc{StochasticBalance} is \emph{opportunistic} in the sense 
that never leaves a request unassigned if it has a neighbor with remaining 
capacity. 

In the following lemma, we first focus on the expected number of matches 
during any round of requests. The lemma not only holds true for 
\textsc{StochasticBalance}, but any opportunistic algorithm.

\begin{lemma}
\label{lem:round}
    Let the random variable $R_i$, $1\leq i \leq n$, denote the number of matches during the assignment of the requests belonging to round $i$ in $G_n^b$. 
    Let $m$ be the number of unused capacity from servers in $\{s_i, s_{i+1}, \ldots, s_n\}$ at the start of round $i$.  It holds that $R_i \sim \min\{\text{Pois}(b), m\}$. More precisely, it holds 
    \[
        \Pr[R_i = k] = \begin{cases} \frac{b^k}{k! e^b} & k < m \,, \\
                            1-\sum_{j=0}^{m-1} \frac{b^j}{j! e^b} & k = m \,. 
                            \end{cases}
    \]
    Moreover, $\EX[R_i] \leq b$.
\end{lemma}

\begin{proof}
    We look at the underlying probability space from the perspective of round $i$. 
    First, recall that there are $b/p$ requests in round $i$, which are all incident to the servers $\{s_i, s_{i+1}, \ldots, s_n\}$. 
    Any opportunistic algorithm assigns each request of round $i$ to one of these servers, as long as there is capacity left. 
    As all edges have the same success probability $p$, one can think of these assignments as a sequence of independent coin tosses,
    where landing on heads means that the assignment succeeds.
    More precisely, we toss a coin that has probability $p$ of landing on heads until there are either $m$ heads or $b/p$ 
    tosses in total. We can simplify this by always tossing the coin $b/p$ times, where we stop counting once we see $m$ heads. 
    Let $R'_i$ be the actual number of heads in the $b/p$ coin tosses. Note that $R_i$ - the number of counted heads - 
    is equal to $\min\{R'_i, m\}$. It follows that $\Pr[R_i = k] = \Pr[R'_i = k]$ if $k < m$ and $\Pr[R_i = m] = \Pr[R'_i \geq m]$. 
    It obviously holds that $R'_i\sim \text{Bin}(b/p, p)$. 
    By the Poisson limit theorem, we have $R'_i \sim \text{Pois}(b)$ as $p\rightarrow 0$.  
    Finally, it holds that $\EX[R_i] \leq \EX[R'_i] = b$.
\end{proof}

With the help of Lemma \ref{lem:round}, we can now upper bound the expected number of successes for each server when executing
\textsc{StochasticBalance} on $G_n^b$. 

\begin{lemma}
\label{lem:sbalserver}
    Let the random variable $S_j$ denote the total number of matches of server $s_j$ after executing \textsc{StochasticBalance} on $G_n^b$. It holds that 
    \[
    \EX[S_j] \leq \min\left\{\sum_{i=1}^j \frac{b}{n-i+1},b \right\} = b\cdot  \min\left\{\sum_{i=1}^j \frac{1}{n-i+1},1 \right\} \,.
    \]
\end{lemma}

\begin{proof}
    Let $S_j^i$ be the number of matches of server $s_j$ during round $i$. 
Observe that all servers $\{s_i, s_{i+1}, \ldots s_n\}$ have the same neighbors until round $i$. 
Therefore, they are treated identically by \textsc{StochasticBalance}. It follows that $\EX[S_j^i]$ is equal for all $i \leq j \leq n$.
Note that $R_i = \sum_{j=i}^n S_j^i$, implying $\EX[S_j^i] = \EX[R_i]/(n-i+1)$. By Lemma \ref{lem:round}, we further get 
$\EX[S_j^i] \leq b/(n-i+1)$, for all $i$. 
Since $s_j$ is a neighbor of all requests from rounds 1 through $j$, we have $S_j = \sum_{i=1}^j S_j^i$. By linearity of 
expectation, $\EX[S_j] = \sum_{i=1}^j \EX[S_j^i] \leq \sum_{i=1}^j \frac{b}{n-i+1}$. 
Moreover, $S_j \leq b$ trivially holds true, which finishes the proof. 
\end{proof}

A reader may recognize the expression $\min\left\{\sum_{i=1}^j \frac{1}{n-i+1},1 \right\}$. It is exactly the load that the server $s_j$ 
is assigned by the well-known (fractional) \textsc{Water-Filling} algorithm on $G_n^1$ with $p=1$. 
Moreover, the same expression also shows up in the proof of the upper bound of $1-1/e$ for the online $b$-matching problem \cite{MSVV}.
Hence it has been shown before that summing up all $\EX\left[S_j\right]$ and dividing by $n\cdot b$ yields $1-1/e$ 
in the limit $n \rightarrow \infty$.
For the sake of completeness, we briefly prove this here as well.

\begin{lemma}
\label{lem:sbal}
    Let the random variable $\textsc{SBal}\left(G_n^b\right)$ denote the number of matches generated by \textsc{StochasticBalance} on $G_n^b$. It holds that 
    \[
        \frac{\EX\left[\textsc{SBal}\left(G_n^b\right)\right]}{n\cdot b} \leq \frac{\left\lceil\left(1-\frac{1}{e}\right)(n+1)\right\rceil \cdot b}{n\cdot b} \xrightarrow{n\rightarrow\infty}1-\frac{1}{e} \,.
    \]
\end{lemma}

\begin{proof}
It holds that $\textsc{SBal}\left(G_n^b\right) = \sum_{j=1}^n S_j$. Lemma \ref{lem:sbalserver} yields 
    \[
        \EX[\textsc{SBal}\left(G_n^b\right)] \leq b\cdot  \sum_{j=1}^n \min\left\{\sum_{i=1}^j \frac{1}{n-i+1},1 \right\} \,.
    \]
    In the following, we will identify an index $k$ depending on $n$, such that $\sum_{i=1}^{k} \frac{1}{n-i+1} \geq 1$.
    Then, it follows that 
    \begin{align*}
        \sum_{j=1}^n \min\left\{\sum_{i=1}^j \frac{1}{n-i+1},1 \right\} &\leq \sum_{j=1}^k \sum_{i=1}^j \frac{1}{n-i+1} + (n-k)\cdot 1 \\
        &\leq \sum_{i=1}^k \sum_{j=i}^k \frac{1}{n-i+1} + (n-k)\cdot \sum_{i=1}^k\frac{1}{n-i+1} \\
        &= \sum_{i=1}^k \frac{(k-i+1)+(n-k)}{n-i+1} = k \,.
    \end{align*}
    Thus, $\EX[\textsc{SBal}] \leq b\cdot k$.

    To find $k$, we lower bound the sum $\sum_{i=1}^{k} \frac{1}{n-i+1}$ by a suitable integral. We have
    \begin{align*}
        \sum_{i=1}^{k} \frac{1}{n-i+1} = \sum_{i=n-k+1}^n \frac{1}{i} \geq \int_{n-k+1}^{n+1} \frac{1}{x} \ \mathrm{d}x = \ln\left(\frac{n+1}{n-k+1} \right) \,.
    \end{align*}
    Solving the inequality $\ln\left(\frac{n+1}{n-k+1} \right) \geq 1$ for $k$ yields 
    \[
        k \geq \left(1-\frac{1}{e}\right)(n+1) \,.
    \]
    Taking $k=\left\lceil\left(1-\frac{1}{e}\right)(n+1)\right\rceil$ finishes the proof.
\end{proof}

\subsection{Stochastic Benchmark on $G_n^b$}
\label{sec:ubSOPT}
In this section, we quantify the stochastic benchmark \textsc{SOpt} on $G_n^b$. Afterward, we compare it to the upper bound on the expected 
number of matches by \textsc{StochasticBalance} to obtain the hardness result. Recall 
that the stochastic benchmark is defined as the expected number of matches of the best possible algorithm that knows the entire graph
including edge probabilities in advance. However, said algorithm does not know which edges will be successful a priori and still 
needs to match the requests according to the arrival order. By definition, the expected number of matches of any algorithm with the restrictions above
lower bounds \textsc{SOpt}.

In the following, we analyze the algorithm \textsc{Greedy} on $G_n^b$. \textsc{Greedy} simply assigns an incoming request to the 
server with the smallest index among all servers with remaining capacity. Note that \textsc{Greedy} needs to know the 
graph in advance to identify the indices of the servers. 
One could show that \textsc{Greedy} is in fact the best 
possible algorithm on $G_n^b$ that obeys the restriction posed by the stochastic benchmark. However, this is not necessary here, as a 
lower bound for \textsc{SOpt} suffices to show an upper bound for $\EX[\textsc{SBal}\left(G_n^b\right)]/\textsc{SOpt}\left(G_n^b\right)$, which is what we are interested in.

To simplify the analysis, we only analyze the exact expected number of successes generated by \textsc{Greedy} on $G_n^1$. 
However, this will suffice to show the same upper bound for all values of $b$. 
Let $T_{n,m}$ be the number of matches by \textsc{Greedy} on $G_n^1$, where only the last $m$ ($\leq n$) servers, i.e. $s_{n-m+1}, \ldots, s_n$, are 
present. Note that we want to determine $\EX\left[T_{n,n} \right]$. We use the recursive nature of $G_n^1$ together with the 
law of total expectation to develop a recurrence relation for $\EX\left[T_{n,m} \right]$.

\begin{lemma}
\label{lem:Trec}
    For $m<n$, it holds that 
    \begin{equation}
    \label{equ:Trec1}
        \EX\left[T_{n,m} \right] = \sum_{k=0}^{m-1}\frac{1}{k!e}\cdot\left(k+\EX\left[T_{n-1,m-k} \right]\right) + \left(1-\sum_{k=0}^{m-1}\frac{1}{k!e}\right)\cdot m \,.
    \end{equation}
    Moreover, for $m=n$, we have 
    \begin{equation}
    \label{equ:Trec2}
        \EX\left[T_{n,n} \right]= \frac{1}{e} \cdot \EX\left[T_{n-1,n-1} \right] + \sum_{k=1}^{n-1}\frac{1}{k!e}\cdot\left(k+\EX\left[T_{n-1,n-k} \right]\right) + \left(1-\sum_{k=0}^{n-1}\frac{1}{k!e}\right)\cdot n \,.
    \end{equation}
\end{lemma}

\begin{proof}
    Consider the number of successes $R_1$ during the first round of requests. 
    As \textsc{Greedy} is an opportunistic algorithm, Lemma \ref{lem:round} can also be applied here. 
    The law of total expectation then yields 
    \begin{equation}
    \label{equ:totalEx}
        \EX\left[T_{n,m} \right] = \sum_{k=0}^m \EX\left[T_{n,m} \mid R_1 = k\right]\cdot \Pr[R_1 = k] \,.
    \end{equation}

    Assume that only the $m$ servers $s_{n-m+1}, \ldots, s_n$ were present at the beginning. 
    Then, if there are $k$ matches during round $1$, we know that the $k$ servers $s_{n-m+1}, \ldots, s_{n-m+k}$ 
    have been matched, by the definition of \textsc{Greedy}.
    Thus, observe that the remaining rounds are identical to the execution of \textsc{Greedy} on $G_{n-1}^1$ where 
    only the last $m-k$ servers are present. Therefore, we are inclined to write 
    $\EX\left[T_{n,m} \mid R_1 = k\right] = k+\EX\left[T_{n-1,m-k}\right]$. However, note that there is an edge case if $m=n$ and $k=0$. 
    If we start with the full graph $G_n^1$ and there are no matches during round $1$, the remaining 
    rounds are identical to $G_{n-1}^1$. Therefore, the precise identity is 
    \[
        \EX\left[T_{n,m} \mid R_1 = k\right] = k+\EX\left[T_{n-1,\min\{n-1,m-k\}}\right] \,.
    \]
    Plugging this into (\ref{equ:totalEx}) and using the probability distribution of $R_1$ according to Lemma \ref{lem:round} 
    finishes the proof. 
\end{proof}

The following lemma solves the recurrence relation. 

\begin{lemma}
\label{lem:Tsolved}
    For all $m\leq n$, it holds that
    \begin{equation*}
       \EX\left[T_{n,m} \right] =  m-\sum_{k=0}^{m-1} \frac{(n-k)^{m-k-1}}{(m-k-1)!e^{n-k}} \,. 
    \end{equation*}
\end{lemma}

\begin{proof}
    By induction over $n$. For $n=0$, $\EX\left[T_{0,0} \right] = 0$ is vacuously true.
    Hence assume that the formula for $\EX\left[T_{n,m}\right]$ holds for $n$ and all $m\leq n$. Now, consider $n+1$. If $m < n+1$, we have according to (\ref{equ:Trec1}) 
    \begin{align*}
        \EX\left[T_{n+1,m} \right] &= \sum_{k=0}^{m-1}\frac{1}{k!e}\cdot\left(k+\EX\left[T_{n,m-k} \right]\right) + \left(1-\sum_{k=0}^{m-1}\frac{1}{k!e}\right)\cdot m \\ 
        & \overset{IH.}{=} \sum_{k=0}^{m-1}\frac{1}{k!e}\cdot\left(k+\left( 
        m-k-\sum_{i=0}^{m-k-1} \frac{(n-i)^{m-k-i-1}}{(m-k-i-1)!e^{n-i}}
        \right)\right) + \left(1-\sum_{k=0}^{m-1}\frac{1}{k!e}\right)\cdot m \\
        &= m-\sum_{k=0}^{m-1}\frac{1}{k!e} \sum_{i=0}^{m-k-1} \frac{(n-i)^{m-k-i-1}}{(m-k-i-1)!e^{n-i}} \,.
    \end{align*}
    To simplify the double sum, we first change the order of summation. $i \leq m-k-1$ is equivalent to $k \leq m-i-1$. 
    Moreover, as $k\geq 0$, $i\leq m-1$. Thus, it holds
    \begin{align*}
        \sum_{k=0}^{m-1}\frac{1}{k!e} \sum_{i=0}^{m-k-1} \frac{(n-i)^{m-k-i-1}}{(m-k-i-1)!e^{n-i}} &= \sum_{i=0}^{m-1}\frac{1}{e^{n+1-i}}  \sum_{k=0}^{m-i-1} \frac{(n-i)^{m-k-i-1}}{k!(m-k-i-1)!}  \\
        &= \sum_{i=0}^{m-1}\frac{1}{(m-i-1)!e^{n+1-i}} \sum_{k=0}^{m-i-1} \binom{m-i-1}{k}(n-i)^{m-k-i-1}  \,.
    \end{align*}
    Now, note that the right sum is the binomial expansion of $(n-i+1)^{m-i-1}$.
    Putting everything together and renaming $i$ to $k$ finish the case $m<n+1$. 

    Similarly, we have for $m=n+1$ by (\ref{equ:Trec2}) 
    \begin{align}
        \EX\left[T_{n+1,n+1} \right] &= \frac{1}{e} \cdot \EX\left[T_{n,n} \right] + \sum_{k=1}^{n}\frac{1}{k!e}\cdot\left(k+\EX\left[T_{n,n+1-k} \right]\right) + \left(1-\sum_{k=0}^{n}\frac{1}{k!e}\right)\cdot (n+1) 
        \nonumber \\ 
        & \overset{IH.}{=} \frac{1}{e}\cdot\left( 
        n-\sum_{k=0}^{n-1} \frac{(n-k)^{n-k-1}}{(n-k-1)!e^{n-k}} 
        \right) + \nonumber \\ 
        & \quad \quad \sum_{k=1}^{n}\frac{1}{k!e}\cdot\left(k+\left( 
        (n+1)-k-\sum_{i=0}^{n-k} \frac{(n-i)^{n-k-i}}{(n-k-i)!e^{n-i}}
        \right)\right) +  \left(1-\sum_{k=0}^{n}\frac{1}{k!e}\right)\cdot (n+1) \nonumber \\
        &= (n+1)-\frac{1}{e}-\sum_{k=0}^{n-1} \frac{(n-k)^{n-k-1}}{(n-k-1)!e^{n+1-k}} -\sum_{k=1}^{n}\frac{1}{k!e} \sum_{i=0}^{n-k} \frac{(n-i)^{n-k-i}}{(n-k-i)!e^{n-i}} \label{equ:temp}\,.
    \end{align}
    
    We again simplify the double sum first by changing the order of summation. $i \leq n-k$ is equivalent to $k \leq n-i$ and $k\geq 1$ implies $i\leq n-1$. Therefore, 
    \begin{align*}
        \sum_{k=1}^{n}\frac{1}{k!e} \sum_{i=0}^{n-k} \frac{(n-i)^{n-k-i}}{(n-k-i)!e^{n-i}} &= \sum_{i=0}^{n-1}\frac{1}{e^{n+1-i}} \sum_{k=1}^{n-i} \frac{(n-i)^{n-k-i}}{k!(n-k-i)!}  \\
        &= \sum_{i=0}^{n-1}\frac{1}{(n-i)!e^{n+1-i}} \sum_{k=1}^{n-i} \binom{n-i}{k}(n-i)^{n-k-i}  \,.
    \end{align*}

    Next, consider the middle sum of (\ref{equ:temp}). We can rewrite it as
    \[
    \sum_{k=0}^{n-1} \frac{(n-k)^{n-k-1}}{(n-k-1)!e^{n+1-k}} = \sum_{k=0}^{n-1} \frac{(n-k)^{n-k}}{(n-k)!e^{n+1-k}} = \sum_{i=0}^{n-1} \frac{1}{(n-i)!e^{n+1-i}}\binom{n-i}{0} (n-i)^{n-i} \,.
    \]
    Putting everything together, we obtain
    \begin{align*}
        \EX\left[T_{n+1,n+1} \right] 
        &= (n+1)-\frac{1}{e}-\sum_{k=0}^{n-1} \frac{(n-k)^{n-k-1}}{(n-k-1)!e^{n+1-k}} -\sum_{k=1}^{n}\frac{1}{k!e} \sum_{i=0}^{n-k} \frac{(n-i)^{n-k-i}}{(n-k-i)!e^{n-i}} \\ 
        &= (n+1)-\frac{1}{e} - \sum_{i=0}^{n-1}\frac{1}{(n-i)!e^{n+1-i}} \sum_{k=0}^{n-i} \binom{n-i}{k}(n-i)^{n-k-i} \\
        &= (n+1)-\frac{1}{e} - \sum_{i=0}^{n-1}\frac{(n+1-i)^{n-i}}{(n-i)!e^{n+1-i}} \\ 
         &= (n+1) - \sum_{i=0}^{n}\frac{(n+1-i)^{n-i}}{(n-i)!e^{n+1-i}}\,.
    \end{align*}
\end{proof}

\begin{lemma}
\label{lem:Gre}
    Let the random variable $\textsc{Gre}\left(G_n^1\right)$ denote the number of matches generated by \textsc{Greedy} on $G_n^1$. It holds that 
    \[
        \frac{\EX\left[\textsc{Gre}\left(G_n^1\right)\right]}{n} = 1- \frac{1}{n}\sum_{k=1}^n \frac{k^{k-1}}{(k-1)!e^k} \xrightarrow{n\rightarrow\infty}1 \,.
    \]
\end{lemma}

\begin{proof}
    Recall that $\textsc{Gre}\left(G_n^1\right) = T_{n,n}$. The equality then immediately follows from Lemma~\ref{lem:Tsolved} and a change of 
    summation order. Therefore, all that is left to show is that 
    \[
        \frac{1}{n}\sum_{k=1}^n \frac{k^{k-1}}{(k-1)!e^k} \xrightarrow{n\rightarrow\infty} 0 \,.
    \]

    By the Cesàro Mean Theorem, it suffices to show that $\frac{k^{k-1}}{(k-1)!e^k} \xrightarrow{k\rightarrow\infty} 0$. 
    Using the Stirling approximation $n! \sim \sqrt{2\pi n}\frac{n^n}{e^n}$, we have
    \[
    \lim_{k\rightarrow \infty} \frac{k^{k-1}}{(k-1)!e^k} = \lim_{k\rightarrow \infty} \frac{k^{k}}{k!e^k} = \lim_{k\rightarrow \infty} \frac{1}{\sqrt{2 \pi k}} = 0 \,.
    \qedhere
    \]
\end{proof}

\begin{theorem}
    No (randomized) algorithm achieves a competitive ratio greater than $1-1/e$ against \textsc{SOpt} for the online $b$-matching problem with stochastic rewards, for all $b$, even for equal and vanishing edge probabilities. 
\end{theorem}

\begin{proof}
    By Lemma \ref{lem:MP}, \textsc{StochasticBalance} is optimal for all graphs $G_n^b$. 
    To lower bound \textsc{SOpt} on $G_n^b$, we perform the standard vertex-splitting reduction of the $b$-matching problem. 
    Recall that we replace every server $s_i$ with capacity $b$ by $b$ servers $s_i^1,s_i^2 \ldots, s_i^b$ with 
    unit capacity. Moreover, each server $s_i^j$ has the same neighbors as the original server $s_i$, where all edge probabilities 
    are still $p \rightarrow 0$.
    Call the resulting graph $G'_{n\cdot b}$. 
    As the fundamental structure of the graph remains the same, we can 
    translate any algorithm on $G_n^b$ into an equivalent algorithm on 
    $G'_{n\cdot b}$ and vice versa. Thus, we have 
    $\textsc{SOpt}\left(G_{n}^b\right) = \textsc{SOpt}\left(G'_{n\cdot b}\right)$. 
    Observe that $G'_{n\cdot b}$ has $n\cdot b$ servers and $n\cdot b/p$ requests. 
    In fact, the graph $G_{n\cdot b}^1$ is a subgraph of $G'_{n\cdot b}$. As any algorithm on $G'_{n\cdot b}$ can easily restrict 
    itself to only those edges present in $G_{n\cdot b}^1$, it holds that $\textsc{SOpt}\left(G'_{n\cdot b}\right) \geq \textsc{SOpt}\left(G_{n\cdot b}^1\right)$.
    Therefore, no (randomized) algorithm can achieve a competitive ratio greater than 
    \[
        \frac{\EX\left[\textsc{SBal}\left(G_n^b\right)\right]}{\textsc{SOpt}\left(G_{n}^b\right)} \leq \frac{\EX\left[\textsc{SBal}\left(G_n^b\right)\right]}{\textsc{SOpt}\left(G_{n\cdot b}^1\right)} \leq  \frac{\EX\left[\textsc{SBal}\left(G_n^b\right)\right]}{\EX\left[\textsc{Gre}\left(G_{n\cdot b}^1\right)\right]} =  \frac{\EX\left[\textsc{SBal}\left(G_n^b\right)\right]}{n\cdot b} \cdot \frac{n \cdot b}{\EX\left[\textsc{Gre}\left(G_{n\cdot b}^1\right)\right]} \,.
    \]
    Finally, Lemma \ref{lem:sbal} implies that the first fraction approaches $1-1/e$, while Lemma~\ref{lem:Gre} yields that the second 
    fraction approaches $1$, both for $n\rightarrow \infty$.
\end{proof}

Note that the upper bound of $1-1/e$ is as small as possible, even for $b=1$, as Goyal and Udwani \cite{GU}
show that \textsc{PerturbedGreedy} achieves this competitiveness against $\textsc{SOpt}$. 

\begin{corollary}
     No (randomized) algorithm achieves a competitive ratio greater than $1-1/e$ against \textsc{Opt} for the online $b$-matching problem with stochastic rewards, for all $b$, even for equal and vanishing edge probabilities. 
\end{corollary}

\section{Lower Bound: Competitiveness of \textsc{StochasticBalance} for $b\rightarrow \infty$}
\label{sec:bal}
In this section, we analyze the \textsc{StochasticBalance} algorithm for large server capacities using the standard 
primal-dual framework by Devanur et al.~\cite{DJK}. 
Huang and Zhang \cite{HZ} were the first to 
successfully apply the primal-dual framework to online matching with stochastic rewards. They focus on the case $b=1$ with vanishing probabilities and show that 
the configuration LP has to be used for the analysis to obtain a non-trivial competitiveness for \textsc{StochasticBalance}. 
However, it turns out that the standard matching LP suffices to show the best possible competitive ratio of $1-1/e$ for 
$b\rightarrow \infty$, even for non-vanishing probabilities.

For any $s\in S$, let \emph{load} $l_s$ of a server $s$ denote
the sum of edge probabilities of requests assigned to $s$. 
If all edge probabilities are identical, {\sc StochasticBalance} simply assigns incoming requests to a neighbor with remaining capacity that has minimum load.
For arbitrary edge probabilities, the algorithm is generalized with the help 
of a non-decreasing function $f$, which is determined during the analysis to optimize the competitive ratio. 

\begin{algorithm}[H]
  \caption{Generalized {\sc StochasticBalance}}
  \label{alg:StochasticBalance}
  \SetAlgoLined
	\While{a new request $r\in R$ arrives} {
	Let $N(r)$ denote the set of neighbors of $r$ with remaining capacity\;
	 \eIf{$N(r)= \emptyset$}{
	  Do not assign $r$\;}{
	  assign $r$ to $\argmax \{p_{s,r}(1-f(l_s)) : s\in N(r)\}$ (break ties arbitrarily)\;
	 }
	}
\end{algorithm}

In the following, we conduct a primal-dual analysis of \textsc{StochasticBalance}. We obtain a competitive ratio of $1-1/e$ against 
the non-stochastic benchmark $\textsc{Opt}$. The result holds true for arbitrary, non-vanishing edge probabilities, if $b\rightarrow \infty$. 
Since the stochastic benchmark is easier than the non-stochastic one, the same competitive ratio is also achieved against $\textsc{SOpt}$. 
Afterward, we outline the changes that are necessary to extend this result to the more general vertex-weighted variant of the problem, 
where each server $s$ moreover has an individual server capacity $b_s$.
\begin{theorem}
\label{theo:SBalSOPT}
    \textsc{StochasticBalance} achieves a competitive ratio of $1-1/e$ for the online $b$-matching problem with stochastic rewards for arbitrary edge probabilities
    against both benchmarks, if $b\rightarrow \infty$.
\end{theorem}

\begin{theorem}
\label{theo:ext}
    \textsc{StochasticBalance} achieves a competitive ratio of $1-1/e$ for the vertex-weighted online $b$-matching problem with stochastic rewards for arbitrary edge probabilities and individual server capacities against both benchmarks,
     if $b_{\min} := \min_{s\in S} b_s \rightarrow \infty$.
\end{theorem}
In the remainder of this section we establish Theorem~\ref{theo:SBalSOPT}. 
Theorem~\ref{theo:ext} will be shown in Appendix~\ref{app:ext}. 

First, consider the standard (relaxed) primal and dual LP of online $b$-matching with stochastic rewards given below. 
Here, the primal variable $m(s,r)$ for each edge indicates the \emph{probability} that $e=\{s,r\}\in E$ is chosen 
by the algorithm (irrespective of if the assignment succeeded or not). 
Note that these probabilities are generally different from $0$ or $1$, even for 
deterministic algorithms, as the randomness of assignments succeeding can influence the decisions of the algorithm.
In the primal LP, the first set of constraints ensures that the server capacities are observed in expectation. 
The second set of constraints ensures that each request is assigned to at most one server. 
Note that the primal program also corresponds to the budgeted allocation problem (cf.\ the discussion on benchmarks in the introduction). As 
mentioned before, one can show that the expected number of matches of any online algorithm is upper bounded 
by the optimal solution to the primal program, which is exactly the non-stochastic benchmark. 
\begin{align*}
\textbf{P: }\text{max} \ &\sum_{\{s,r\}\in E}p_{s,r} \cdot m(s,r) 
& \textbf{D: }\text{min} \ &b\cdot\sum_{s\in S}  x(s) + \sum_{r\in R} y(r) \\
\text{s.t.} \ &\sum_{r:\{s,r\}\in E} p_{s,r} \cdot m(s,r) \leq b, \ (\forall s \in S)
&\text{s.t.} \ & p_{s,r} \cdot x(s)+y(r) \geq p_{s,r} , \ (\forall \{s,r\}\in E) \\
&\sum_{s:\{s,r\}\in E} m(s,r) \leq 1, \ (\forall r\in R)
& &x(s),\,y(r)\geq 0, \ (\forall s\in S, \forall r\in R)\, \\
& m(s,r) \geq 0, \ (\forall \{s,r\}\in E) \,
\end{align*}

Now, we can explain the primal-dual framework.
We construct a solution to the primal program by setting the primal variables according to \textsc{StochasticBalance}. 
The value of the primal solution is thus the expected number of matches generated by \textsc{StochasticBalance}. 
In parallel, we create a feasible solution to the dual program with the value of the primal solution multiplied by a constant $1/c$.
It then follows by weak duality that \textsc{StochasticBalance} is $c$ competitive against the optimal solution to the primal 
program, i.e. the non-stochastic benchmark.  

More specifically, depending on the realization of edge successes and failures, we maintain a (random) assignment 
to the variables $\hat{m}(s,r)$, $\hat{x}(s)$ and $\hat{y}(r)$. 
Let the random variables $P$ and $D$ denote the value of the random primal and dual solutions, respectively. We denote a change in the values $P$ and $D$ by $\Delta P$ and $\Delta D$, respectively. 
Moreover, we define $m(s,r) := \EX[\hat{m}(s,r)]$, 
$x(s) := \EX[\hat{x}(s)]$ and $y(r) := \EX[\hat{y}(r)]$, where expectation is taken over the randomness of edges succeeding or failing. 
Initially, all variables are 0. Whenever \textsc{StochasticBalance} assigns a request $r$ to a server $s$, 
we set $\hat{m}(s,r)$ to 1. Moreover, we increase $\hat{x}(s)$ by and set $\hat{y}(r)$ to 
\[
    \Delta \hat{x}(s) = p_{s,r}\cdot\frac{f(l_s)}{b\cdot c} \quad \text{and} \quad \hat{y}(r) = p_{s,r}\cdot \frac{1-f(l_s)}{c} \,,
\]
respectively. 
Here, $l_s$ denotes the load of server $s$ \emph{before} the assignment.  Note that it always holds that $\Delta P/c = p_{s,r}/c= \Delta D$. Thus, summing over all assignments of 
the algorithm and taking expectation yields $\EX[P]/c = \EX[D]$. 
This is exactly the desired relation between the values of the primal and dual solutions to show $c$-competitiveness. 
All that is left to show now is that $f$ and $c$ can be chosen in a way that guarantees the feasibility of the dual solution. 

Since we are primarily interested in the best possible competitive ratio $c$ as $b\rightarrow \infty$, we will refrain 
from optimizing $f$ and $c$ for an arbitrary value of $b$. Even for vanishing probabilities, this involves solving integral equations using 
numeric methods \cite{HZ, HJSS}. 
Instead, we use the simple combination of $c=1-1/e$ and
\[
f(x) = \begin{cases}
    e^{\frac{x}{b}-1} &  x \leq b \,, \\
    1 &   x > b \,.
\end{cases}  
\]
We argue that this yields a feasible solution in the limit $b\rightarrow \infty$. By the choice of $f$, we always have $\hat{x}(s)\geq 0$ and $\hat{y}(r)\geq 0$. It immediately follows that $x(s) \geq 0$ and $y(r)\geq 0$. 
Hence we only need to show that 
our dual solution satisfies the first set of constraints. 

The following lemma allows us to lower bound $\hat{x}(s)$ by only considering
the load of a server $s$. This will simplify the analysis later on, since 
we then do not need to consider the exact set of edges assigned to $s$ anymore. 

\begin{lemma}
\label{lem:x(s)Integral}
    For any load $l_s$ of $s$, where $\{p_1, p_2, \ldots, p_j\}$ with $\sum_{i=1}^j p_i = l_s$ are the individual probabilities of the edges assigned to $s$, it holds that 
    \[
    \hat{x}(s) = \frac{1}{b\cdot c} \cdot \sum_{i=0}^{j-1} f \left(\sum_{l=1}^i p_l\right)\cdot p_{i+1} \geq \frac{1}{b\cdot c}\int_{-1}^{l_s-1} f(x) \ \mathrm{d}x \,.
    \]
\end{lemma}

\begin{proof}
    The first equality follows from our update rule for $\hat{x}(s)$. Recall that $f$ is a non-decreasing function.
    Consider Figure \ref{fig:integral} for a visualization of the sum above. The individual summands can be interpreted as areas of rectangles with 
    height $f\left(\sum_{l=1}^i p_l\right)$ and width $p_{i+1}$. 
    If we fix the left edge of the rectangle with width $p_{i+1}$ at $x=\sum_{l=1}^i p_l$, we see that all rectangles are positioned under the graph of $f$. 
    Therefore, it holds that 
    \[
        \int_{0}^{l_s} f(x) \ \mathrm{d}x \geq \sum_{i=0}^{j-1} f \left(\sum_{l=1}^i p_l\right)\cdot p_{i+1} \,.
    \]
    However, we will need a lower bound of $\hat{x}(s)$ for the subsequent analysis. 
    For this, we move all rectangles one unit to the left. Since the width of a 
    rectangle is at most $1$ and $f$ is non-decreasing, it now holds that all rectangles are positioned above the graph of $f$. The lemma follows. 
\end{proof}

\begin{figure}
    \centering
    \begin{tikzpicture}
  \draw[->] (-1, 0) -- (7, 0) node[right] {$x$};
  \draw[->] (0, -0.5) -- (0, 3) node[above] {$y$};
  \draw[domain=-1:7, smooth, variable=\x, thick] plot ({\x}, {(2.72)^(\x/7)}) node[above] {$f(x)$};
  \filldraw[fill=gray!50] (0,0) rectangle ++(1,1);
  \filldraw[fill=gray!50] (1,0) rectangle ++(0.6,{2.72^(1/7)});
  \filldraw[fill=gray!50] (1.6,0) rectangle ++(0.9,{2.72^(1.6/7)});

  \filldraw[fill=gray!50] (5.5,0) rectangle ++(0.5,{2.72^(5.5/7)});
  \draw[dashed] (6,-0.5) node[below right] {$l_s$} -- (6,3) ;

  \node at (4,0.7) {$\dots$} ;
  \node at (-0.25,-0.25) {$0$} ;

  \draw [decorate,decoration={brace,amplitude=6pt},xshift=0pt,yshift=-2pt]
(1,0) -- (0,0) node [black,midway,yshift=-0.5cm] {$p_1$};
\draw [decorate,decoration={brace,amplitude=6pt},xshift=0pt,yshift=-2pt]
(1.6,0) -- (1,0) node [black,midway,yshift=-0.5cm] {$p_2$};
\draw [decorate,decoration={brace,amplitude=6pt},xshift=0pt,yshift=-2pt]
(2.5,0) -- (1.6,0) node [black,midway,yshift=-0.5cm] {$p_3$};
\draw [decorate,decoration={brace,amplitude=6pt},xshift=0pt,yshift=-2pt]
(6,0) -- (5.5,0) node [black,midway,yshift=-0.5cm] {$p_j$};
\end{tikzpicture}
    \caption{The area of the gray rectangles is proportional to the value of $\hat{x}(s)$ when $s$ is assigned edges with probabilities $p_1, p_2, \ldots, p_j$.}
    \label{fig:integral}
\end{figure}
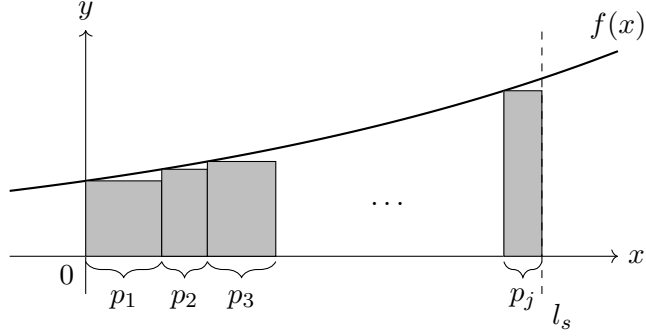

Let $E_s$ denote the set of edges incident to $s$.
Similar to the other primal-dual analyses of randomized online matching algorithms, we will focus on a single server $s$ and fix 
the randomness of all other servers. 
More precisely, we fix the outcome, i.e. success or failure, of every edge in $E\setminus E_s$. Formally, we define the binary vector $Z$ that has 
one entry $Z(e)$ for every edge $e\in E$. $Z(e) = 1$ then means that the 
edge $e$ succeeds when it is chosen by the algorithm. 
We use $Z_{-s}$ and $Z_s$ to denote the parts of $Z$ that contain all the entries for edges $E\setminus E_s$ and $E_s$, respectively. 
We will show that 
 \[
 \EX_{Z_s}\left[p_{s,r} \cdot \hat{x}(s) + \hat{y}(r) \mid Z_{-s} \right] \geq p_{s,r} \cdot (1-\varepsilon) \,,
 \]
 with $\varepsilon \rightarrow 0$ as $b\rightarrow \infty$. 
 Taking expectation over $Z_{-s}$ will then show the feasibility of the dual solution for large server capacities. 
 
\begin{lemma}
\label{lem:IdenticalExecutions}
    For any realization of the random outcomes $Z_{-s}$ of all edges in $E\setminus E_s$, 
    the assignment produced by \textsc{StochasticBalance} for any realization $Z_s$ is identical to the assignment where all edges 
    in $E_s$ are unsuccessful, i.e. $Z_s = \Vec{0}$, up until the 
    point in time when $s$ has no capacity left.
\end{lemma}

\begin{proof}
    This follows from the fact that \textsc{StochasticBalance} is indifferent to how many matches the server $s$ has collected so far.
    It only considers the load of $s$ as long as $s$ has remaining capacity. 
    Formally, we prove this by induction over the number of requests assigned by \textsc{StochasticBalance}. Initially, the statement 
    is trivially true. Thus, consider that the assignments are identical so 
    far and that $s$ still has remaining capacity. It immediately follows that the loads of all servers are also identical in both cases. 
    Moreover, the set of servers with remaining capacities are also identical. To see this, recall that only the outcomes of edges in $E_s$ are 
    different between the two executions. This implies that all servers but $s$ have the same number of successful matches in both cases. 
    Thus, as long as $s$ has remaining capacity for any $Z_s$, all properties relevant to \textsc{StochasticBalance} 
    are identical to the execution with $Z_s = \Vec{0}$. Since \textsc{StochasticBalance} is a deterministic algorithm, 
    it assigns the next incoming request to the same server in both cases. 
\end{proof} 

\begin{lemma}
\label{lem:AssignmentsOfS}
    Given the random outcomes $Z_{-s}$ of all edges in $E\setminus E_s$, let $R_s = \{r_1, r_2, \ldots, r_k\}$ denote the set of requests assigned to 
    $s$ ordered by arrival times if $Z_s = \Vec{0}$. Let
    the random variable $L_s$ denote the index of the request (in $R_s$) that results in the $b$-th success of $s$. 
    If there are less than $b$ matches in total, we set $L_s = k+1$. \textsc{StochasticBalance} assigns the requests $\{r_1,r_2,\ldots, r_j\}$, $j<k$, to $s$ if and only 
    if $L_s = j$. Similarly, $s$ is assigned all requests from $R_s$ if and only if $L_s \geq k$. 
\end{lemma}

\begin{proof}
    This immediately follows from Lemma \ref{lem:IdenticalExecutions}. 
    As long as $s$ has remaining capacity, \textsc{StochasticBalance} assigns all requests as if $Z_s = \Vec{0}$. 
    Hence $s$ is only assigned requests in $R_s$. Moreover, if the assignment of $r_j$ produces the $b$-th success of $s$, 
    $s$ is obviously assigned no further requests since it has no capacity left. 
\end{proof}
    
With the help of Lemma \ref{lem:AssignmentsOfS}, we can now lower bound the expected values of $\hat{x}(s)$ and $\hat{y}(r)$ for all neighbors $r$ of $s$, 
given $Z_{-s}$. To simplify the notation, we define $p_j := p_{s,r_j}$, $l_s^j := \sum_{i=0}^j p_i$ and
\[
    I_j := \frac{1}{b} \cdot \int_{-1}^{l_s^j-1} f(x) \ \mathrm{d} x \,,
\]
for all $j\in [k]$. 
\begin{lemma}
    \label{lem:x(s)}
    For any realization of the random outcomes $Z_{-s}$ of all edges in $E\setminus E_s$ and the corresponding set of requests $R_s$, it holds that 
    \[
    \EX_{Z_s}\left[\hat{x}(s) \mid Z_{-s} \right] = \frac{1}{c} \cdot \left(\sum_{j=0}^{k-1} \left(I_{j+1}-I_j\right) \cdot \left(1-\Pr\left[L_s  \leq j\right] \right) \right) \,,
    \]
    where we define $I_0 := 0$. 
\end{lemma}

\begin{proof}
    By Lemma \ref{lem:x(s)Integral}, the value of $\hat{x}(s)$ is at least $I_j/c$ when the requests $\{r_1,r_2,\ldots,r_j\}$ are assigned to $s$. 
    By Lemma \ref{lem:AssignmentsOfS}, it further follows that 
    \[
        \EX_{Z_s}\left[\hat{x}(s) \mid Z_{-s} \right] = \frac{1}{c} \cdot \left(\sum_{j=1}^{k-1} \Pr\left[L_s = j\right] \cdot I_j + \Pr\left[L_s \geq k\right] \cdot I_k \right)\,.
    \]
    Next, we use \emph{summation by parts} to rearrange the sum in the expression above. 

    \begin{fact}[Summation by Parts]
    \label{fac:sbp}
        Let $\left( f_n \right)$ and $\left( g_n \right)$ be two sequences. It then holds that 
        \[
        \sum_{k=m}^n \left(f_{k+1}-f_k\right)g_k= \left(f_{n+1}g_{n+1}-f_mg_m \right) - \sum_{k=m}^n f_{k+1}\left(g_{k+1} - g_k\right) \,.
        \]
    \end{fact}
    Note that $\Pr\left[L_s = j\right] = \Pr\left[L_s \leq j\right]-\Pr\left[L_s \leq j-1\right]$. It follows that 
    \begin{align*}
        \sum_{j=1}^{k-1} \Pr\left[L_s = j\right] \cdot I_j &= \sum_{j=1}^{k-1} \left(\Pr\left[L_s \leq j\right]-\Pr\left[L_s \leq j-1\right]\right) \cdot I_j \\
        &= \Pr\left[L_s \leq k-1\right] \cdot I_k - \Pr\left[L_s \leq 0 \right] \cdot I_1 
         - \sum_{j=1}^{k-1} \Pr\left[L_s \leq j\right] \cdot \left(I_{j+1}-I_j \right) \,,
    \end{align*}
    where we used Fact \ref{fac:sbp} with $f_j=\Pr\left[ L_s \leq j-1\right]$ and $g_j = I_j$. Observe that $\Pr\left[L_s \leq 0 \right] = 0$. 
    Putting everything together yields
    \[
    \EX_{Z_s}\left[\hat{x}(s) \mid Z_{-s} \right] = \frac{1}{c} \cdot \left(I_k - \sum_{j=1}^{k-1} \Pr\left[L_s \leq j\right] \cdot \left(I_{j+1}-I_j \right) \right) \,.
    \]
    Finally, we use $I_k = \sum_{j=0}^{k-1} \left(I_{j+1}-I_j\right)$ and let the sum above start at $j=0$. Again, as 
    $\Pr\left[L_s \leq 0 \right] = 0$, this does not change the result. 
    This finishes the proof.
\end{proof}

\begin{lemma}
    \label{lem:y(r)}
    Consider any edge $\{s,r\}\in E$.
    For any realization of the random outcomes $Z_{-s}$ of all edges in $E\setminus E_s$ and the corresponding set of requests $R_s$, it holds that 
    \[
    \EX_{Z_s}\left[\hat{y}(r) \mid Z_{-s} \right] \geq \left(1-\Pr\left[L_s \leq k\right]\right) \cdot \frac{p_{s,r}}{c} \cdot \left( 1- f\left(l_s^k \right) \right) \,.
    \]
\end{lemma}

\begin{proof}
    Let $l_s^r$ denote the load of $s$ when $r$ arrived. Recall that $\hat{y}(r)$ is set to $p_{s',r}\cdot\left(1-f\left(l_{s'}^r\right)\right)/c$, 
    when $r$ is assigned to some server $s'$. 
    Since $s$ is never assigned any requests not in $R_s$, 
    we obviously have $l_s^r \leq l_s^k$.
    If $s$ is available throughout the algorithm, \textsc{StochasticBalance} considers $s$ when making its decision for $r$. 
    Hence $r$ is assigned to some server $s'$ with 
    \[
    p_{s',r} \left(1-f\left(l_{s'}^r \right) \right) \geq p_{s,r} \left(1-f\left(l_{s}^r \right) \right) \geq p_{s,r} \left(1-f\left(l_{s}^k\right) \right) \,, 
    \]
    where the last inequality is true because $f$ is non-decreasing.
    By Lemma \ref{lem:IdenticalExecutions}, $s$ is available throughout the algorithm if $s$ has less than $b$ successful matches 
    from the assignments in $R_s$. By definition, this happens with probability $\Pr\left[L_s = k+1\right] = 1-\Pr\left[L_s \leq k\right]$. 
    If $s$ is not available throughout the algorithm, we trivially have $\hat{y}(r) \geq 0$. Combining these two cases finishes the proof. 
\end{proof}

The following lemma lets us simplify the probabilities from Lemmas \ref{lem:x(s)} and \ref{lem:y(r)} if $b\rightarrow \infty$. 
We prove it with the help of Chebyshev's inequality. 

\begin{lemma}
\label{lem:probTo1}
    If the load $l_s$ of $s$ is at most $b-\omega\left(\sqrt{b}\right)$, the probability that $s$ has less than $b$ matches converges to $1$ as $b\rightarrow \infty$.
\end{lemma}

\begin{proof}
    Let $\{p_1, p_2, \ldots, p_j\}$ with $\sum_{i=1}^j p_i = l_s$ be the individual probabilities of the edges assigned to $s$.
    Let $X$ be the number of heads in a series of $j$ coin tosses, where the $i$-th coin toss has probability $p_i$ of landing on heads. 
    It holds that the probability that $s$ has less than $b$ matches is exactly $\Pr[X < b]$.
    Note that $X$ has the Poisson binomial distribution with $\EX[X] = l_s$. 
    It is known that the variance of $X$ is maximized if all $p_i = l_s/j$ are equal. 
    This results in a binomial distribution with variance 
    \[
    j\cdot \frac{l_s}{j} \cdot \left(1-\frac{l_s}{j} \right) = l_s \left(1-\frac{l_s}{j}\right) < l_s \,.
    \]
    Using Chebyshev's inequality, we get that 
    \[
        \Pr[X \geq b] \leq \Pr\left[|X - l_s| \geq b-l_s\right] < \frac{l_s}{\left(b-l_s\right)^2} \,. 
    \]
    The function $g(x):= x/(b-x)^2$ is increasing for $x<b$. Therefore, it follows for $l_s \leq b-\omega\left(\sqrt{b}\right)$ 
    \[
        \Pr[X \geq b] < \frac{b-\omega\left(\sqrt{b}\right)}{\left(\omega\left(\sqrt{b}\right)\right)^2} < \frac{b}{\omega\left(b\right)} \xrightarrow{b \rightarrow\infty} 0\,. 
    \]
    This implies $\Pr[X<b] = 1-\Pr[X \geq b] \rightarrow 1$ for $b\rightarrow \infty$.
\end{proof}

Finally, we can combine the previous lemmas to show that our constructed dual solution is feasible in the limit $b\rightarrow \infty$.

\begin{lemma}
\label{lem:condSatisfiability}
    For any realization $Z_{-s}$ of all edges in $E\setminus E_s$, it holds for any edge $\{s,r\}\in E$ 
    \begin{equation*}
 \EX_{Z_s}\left[p_{s,r} \cdot \hat{x}(s) + \hat{y}(r) \mid Z_{-s} \right] \geq p_{s,r} \cdot (1-\varepsilon) \,,
    \end{equation*}
    with $\varepsilon \rightarrow 0$ as $b\rightarrow \infty$. 
\end{lemma}

\begin{proof}
    Combining Lemmas \ref{lem:x(s)} and \ref{lem:y(r)}, we obtain for any edge $\{s,r\}\in E$
\begin{multline*}
    \EX_{Z_s}\left[p_{s,r}\cdot \hat{x}(s) + y(r) \mid Z_{-s} \right] = \\ \frac{p_{s,r}}{c} \cdot \left(\sum_{j=0}^{k-1} \left(I_{j+1}-I_j\right) \cdot \left(1-\Pr\left[L_s  \leq j\right] \right)  + \left( 1-\Pr\left[L_s \leq k\right]\right) \cdot \left( 1- f\left(l_s^k \right) \right)\right) \,.
\end{multline*}
    Therefore, it suffices to show that
    \begin{equation}
    \label{equ:suffice}
        \sum_{j=0}^{k-1} \left(I_{j+1}-I_j\right) \cdot \left(1-\Pr\left[L_s  \leq j\right] \right)  + \left( 1-\Pr\left[L_s \leq k\right]\right) \cdot \left( 1- f\left(l_s^k \right) \right) \geq c\cdot(1-\varepsilon) \,.
    \end{equation}    
    Recall that $R_s = \{r_1, r_2, \ldots, r_k\}$ denotes the set of requests 
    assigned to $s$ if $s$ has no successful matches with $p_j := p_{s,r_j}$ and $l_s^j := \sum_{i=1}^j p_i$, for all $j\in [k]$.
    Moreover, $L_s$ is the index of the request in $R_s$ that results in the $b$-th match of $s$, 
    where we defined $L_s = k+1$ if there are less than $b$ matches in total. Observe that $1-\Pr\left[L_s \leq j\right]$ is equivalent to the probability that 
    there are less than $b$ matches when $s$ is assigned $\{r_1,r_2, \ldots, r_j\}$.
    
    We do a case distinction over $l_s^k$. If $l_s^k \leq b-\sqrt[3]{b^2}$,
    we obviously also have $l_s^j \leq b-b^{2/3}$, for all $j\in [k]$. By Lemma \ref{lem:probTo1}, it 
    then follows that $1-\Pr\left[L_s \leq j\right] \rightarrow 1$ as $b\rightarrow \infty$. Hence we can write $1-\Pr\left[L_s \leq j\right] \geq (1-\varepsilon')$ for a suitable $\varepsilon'$, which converges 
    to $0$ for large $b$. This lets us lower bound the LHS of (\ref{equ:suffice}) by 
    \[
    (1-\varepsilon') \cdot \left( \sum_{j=0}^{k-1} \left(I_{j+1}-I_j\right) +  1- f\left(l_s^k \right)  \right) = (1-\varepsilon') \cdot \left( I_k +   1- f\left(l_s^k \right) \right) \,.
    \]
    Plugging in the definition of $I_k$ and simplifying notation with $l:=l_s^k$, we finally obtain
    \begin{align*}
        (1-\varepsilon') \cdot \left( I_k +   1- f\left(l \right) \right) &= (1-\varepsilon')\cdot \left(\frac{1}{b} \cdot \int_{-1}^{l-1} e^{\frac{x}{b}-1} \ \mathrm{d} x + 1 - e^{\frac{l}{b}-1}  \right) \\ 
        &= (1-\varepsilon')\cdot \left(e^{\frac{l-1}{b}-1} - e^{-\frac{1}{b}-1} + 1 - e^{\frac{l}{b}-1} \right) \\
        &= (1-\varepsilon')\cdot \left(1- \frac{e^{-\frac{1}{b}}}{e} + e^{\frac{l}{b}-1} \cdot \left(e^{-\frac{1}{b}} - 1\right)\right)\,.
    \end{align*}
    Note that the expression in the right bracket converges to $1-1/e=c$ for $b\rightarrow \infty$. This implies that there exists a suitable $\varepsilon$ with $\varepsilon \rightarrow 0$ as $b\rightarrow \infty$, 
    such that (\ref{equ:suffice}) holds.
    
    On the other hand, if $l_s^k>b-\sqrt[3]{b^2}$, define $m$ as the largest index in $[k]$ with $l_s^m \leq b-\sqrt[3]{b^2}$. 
    Similar to the case 
    above, it holds that $1-\Pr\left[L_s \leq j\right] \rightarrow 1$ as $b\rightarrow \infty$, for all $j \in [m]$. Therefore, we 
 again lower bound these probabilities by $(1-\varepsilon')$, where $\varepsilon'$ converges to 0 for large $b$. We obtain the following 
    lower bound for the LHS of (\ref{equ:suffice})
    \[
    (1-\varepsilon')\cdot \left( \sum_{j=0}^{m} \left(I_{j+1}-I_j\right) + 0 \right) = (1-\varepsilon') \cdot I_{m+1} \,.
    \]
    Note that $l_s^{m+1} > b-\sqrt[3]{b^2}$. We simplify the notation by letting $l:=l_s^{m+1}$. We conclude
    \begin{align*}
        (1-\varepsilon') \cdot I_{m+1} &= (1-\varepsilon')\cdot \frac{1}{b} \cdot \int_{-1}^{l-1} e^{\frac{x}{b}-1} \ \mathrm{d} x  \\ 
        &> (1-\varepsilon') \cdot \frac{1}{b} \cdot \int_{0}^{b-\sqrt[3]{b^2}-1} e^{\frac{x}{b}-1} \ \mathrm{d} x  \\
        &= (1-\varepsilon') \cdot \left( e^{-b^{-1/3}-\frac{1}{b}}-\frac{1}{e}\right) \,.
    \end{align*}
    Analogous to before, the expression in the right bracket converges to $1-1/e=c$ for $b\rightarrow \infty$, implying 
    that there exists a suitable $\varepsilon$ with $\varepsilon \rightarrow 0$ as $b\rightarrow \infty$, 
    such that (\ref{equ:suffice}) holds.
\end{proof}

\begin{proof}[Proof of Theorem \ref{theo:SBalSOPT}]
    All that remains to show is that our constructed dual solution is feasible if $b\rightarrow \infty$. Consider any edge $\{s,r\}\in E$.
    Recall that $x(s):=\EX[\hat{x}(s)]$ and $y(r):=\EX[\hat{y}(r)]$. 
    Taking expectation over $Z_{-s}$ on both sides of Lemma \ref{lem:condSatisfiability} yields
 \[
 p_{s,r} \cdot x(s) + y(r)  \geq p_{s,r} \cdot (1-\varepsilon) \,, 
 \]
 by the tower property of conditional expectation. Hence our constructed dual solution is almost feasible.
 Recall that we constructed the solutions such that $\EX[\textsc{SBal}]=\EX[P]=c\cdot \EX[D]$. Moreover, we can obtain a feasible dual solution by
 dividing all dual variables by $(1-\varepsilon)$. The resulting dual solution has value $\EX[D'] = \EX[D]/(1-\varepsilon)$. Therefore, we have
 \[
 \EX[\textsc{SBal}] = \EX[P] = \left(1-\frac{1}{e} \right)\cdot \EX[D] = \left(1-\frac{1}{e} \right)\cdot(1-\varepsilon)\cdot \EX[D'] \geq \left(1-\frac{1}{e} \right)\cdot(1-\varepsilon)\cdot  \textsc{Opt} \,,
 \]
 where the inequality follows from weak duality. This implies a competitiveness of 
 \[
    \left(1-\frac{1}{e} \right)\cdot(1-\varepsilon) \xrightarrow{b\rightarrow\infty} \left(1-\frac{1}{e} \right) \,.
 \]
 Since $\textsc{SOpt} \leq \textsc{Opt}$, we immediately obtain the 
 same competitiveness against the stochastic benchmark.
\end{proof}


\bibliographystyle{plain}
\bibliography{main}

\appendix   

\section{Extensions of \textsc{StochasticBalance}}\label{app:ext}

In this appendix, we describe how the \textsc{StochasticBalance} algorithm and our primal-dual analysis 
can be extended to more general variants of the online $b$-matching problem with stochastic rewards. 
First, in Appendix \ref{app:variableb}, we consider the variation of the problem where every server $s \in S$ has an individual capacity $b_s$. Then, in Appendix \ref{app:weighted}, vertex weights 
are added on the offline side. In combination, the arguments establish Theorem~\ref{theo:ext}.

\subsection{Extension for Variable Server Capacities} \label{app:variableb}

In this section, we sketch the changes to \textsc{StochasticBalance} and our primal-dual analysis that are necessary in order to handle 
individual server capacities. In this setting, each server $s\in S$ can be successfully matched up to $b_s$ times. 
For this, we define a function $f_s$ for each server $s$, where 
\[
    f_s(x) = \begin{cases}
    e^{\frac{x}{b_s}-1} &  x \leq b_s \,, \\
    1 &   x > b_s \,.
\end{cases}  
\]
Then, we change \textsc{StochasticBalance} such that it assigns an incoming request to the 
neighbor a remaining capacity that maximizes $p_{s,r}(1-f_s(l_s))$. 

We proceed with the primal-dual analysis the same way as before. The only 
change is that, whenever \textsc{StochasticBalance} assigns a request $r$ to 
a server $s$, we now increase $\hat{x}(s)$ and set $y(r)$ to 
\[
    \Delta \hat{x}(s) = p_{s,r}\cdot\frac{f_s(l_s)}{b_s\cdot c} \quad \text{and} \quad \hat{y}(r) = p_{s,r}\cdot \frac{1-f_s(l_s)}{c} \,,
\]
respectively. Note that the increase of $p_{s,r}$ in the primal 
assignment is still translated to a gain of $p_{s,r}/c$ in the dual assignment. 
The same arguments as before imply that we only have to show the feasibility of the constructed dual solution to show $c$-competitiveness.

Observe that –- except for Lemma \ref{lem:y(r)} –- we always considered a single server $s$ at the time throughout 
the proof of dual feasibility in Section \ref{sec:bal}.
Hence the analysis can essentially be copied verbatim, replacing $b$ with $b_s$ and $f$ with $f_s$. 
In the proof of Lemma \ref{lem:y(r)}, we argued that the dual variable $\hat{y}(r)$ of any neighbor $r$ of $s$ is lower bound by the value that 
$\hat{y}(r)$ would take on if $r$ was matched to $s$, if $s$ is available throughout the whole algorithm. The same line of reasoning 
still holds here, since \textsc{StochasticBalance} even so chooses 
a neighboring server for $r$ such that the dual variable $\hat{y}(r)$ is maximized.

Overall, we obtain the same competitiveness of $1-1/e$ if the capacity $b_s$ of each server $s\in S$ grows large.
This is equivalent to $b_{\min} := \min_{s\in S} b_s \rightarrow \infty$.

\subsection{Extension for Vertex Weights} \label{app:weighted}

At last, we consider the vertex-weighted extension of the problem. Here, each 
server $s\in S$ further has a weight $w_s$ assigned to it, indicating the 
weight of each edge incident to $s$. The goal is to maximize the expected 
total weight of successful matches. The primal and dual linear programs 
for this problem are given below. 

\begin{align*}
\textbf{P: }\text{max} \ &\sum_{\{s,r\}\in E}p_{s,r} \cdot w_s \cdot m(s,r) 
& \textbf{D: }\text{min} \ &\sum_{s\in S}  b_s \cdot x(s) + \sum_{r\in R} y(r) \\
\text{s.t.} \ &\sum_{r:\{s,r\}\in E} p_{s,r} \cdot m(s,r) \leq b_s, \ (\forall s \in S)
&\text{s.t.} \ & p_{s,r} \cdot x(s)+y(r) \geq p_{s,r} \cdot w_s , \ (\forall \{s,r\}\in E) \\
&\sum_{s:\{s,r\}\in E} m(s,r) \leq 1, \ (\forall r\in R)
& &x(s),\,y(r)\geq 0, \ (\forall s\in S, \forall r\in R)\, \\
& m(s,r) \geq 0, \ (\forall \{s,r\}\in E) \,
\end{align*}

We further adapt \textsc{StochasticBalance} such that it assigns an arriving 
request to a neighbor with remaining capacity that maximizes $w_s\cdot p_{s,r}(1-f_s(l_s))$. 
Note that the definition of the load $l_s$ of a server $s$ remains 
unchanged. 
In the primal-dual analysis, we still set $\hat{m}(s,r)=1$ whenever \textsc{StochasticBalance}
assigns a request $r$ to a server $s$. However, note that this now 
increases the value of the primal assignment by $w_s \cdot p_{s,r}$. Therefore, we can now increase $\hat{x}(s)$ and set $\hat{y}(r)$ to 
\[
    \Delta \hat{x}(s) = w_s \cdot p_{s,r}\cdot\frac{f_s(l_s)}{b_s\cdot c} \quad \text{and} \quad \hat{y}(r) = w_s\cdot p_{s,r}\cdot \frac{1-f_s(l_s)}{c} \,,
\]
respectively. Again, it only remains to prove the feasibility of the 
dual solution by showing
\begin{equation*}
    p_{s,r} \cdot x(s)+y(r) \geq p_{s,r} \cdot w_s , \ \forall \{s,r\}\in E \,.
\end{equation*}

Lemmas \ref{lem:IdenticalExecutions}, \ref{lem:AssignmentsOfS} and 
\ref{lem:probTo1} carry over without any changes. Moreover, it is easy to see 
that both the exact expression and lower bound of Lemma \ref{lem:x(s)Integral}
are now multiplied by $w_s$. It follows that the value of $\EX\left[\hat{x}(s) \mid Z_{-s} \right]$ from Lemma \ref{lem:x(s)}
is also multiplied by $w_s$. For Lemma \ref{lem:y(r)}, notice that 
\textsc{StochasticBalance} still assigns an incoming request $r$ such 
that $\hat{y}(r)$ is maximized. Thus, the same arguments as before can then 
be used to obtain the previous lower bound of $\EX\left[\hat{y}(r) \mid Z_{-s} \right]$ multiplied by $w_s$. 
Carrying on, one can see that Lemma \ref{lem:condSatisfiability} now shows 
\[
 \EX_{Z_s}\left[p_{s,r} \cdot \hat{x}(s) + \hat{y}(r) \mid Z_{-s} \right]\geq w_s \cdot p_{s,r} \cdot (1-\varepsilon) \,.
\]
Dual feasibility and thus Theorem \ref{theo:ext} then follow by taking expectation over $Z_{-s}$.

\end{document}